\documentclass[12pt]{article}

\usepackage{amssymb} 
\usepackage{latexsym} 
\usepackage{amsmath}
\usepackage{color}
\usepackage{epsfig}
\usepackage{graphicx}
\usepackage[all]{xy}
\input xy
\usepackage[mathcal]{eucal}
\usepackage{enumerate}
\usepackage{anysize}
\usepackage[english]{babel}
\usepackage{bbm}
\usepackage{hyperref}

\newtheorem{definition}{Definition}

\newtheorem{remark}{Remark}
\newtheorem{prop}{Proposition}
\newtheorem{theorem}{Theorem}

\newtheorem{lemma}{Lemma}

\newtheorem{corollary}{Corollary}

\newtheorem{proofTh} {Proof.}

\newcommand{\CVD} {\hspace*{\fill}$\Box$}

\newenvironment{proof}{\begin{proofTh}\em}{\CVD\end{proofTh}}

\def\qed{\hfill\raisebox{3pt}{\fbox{\rule{0mm}{1mm}\hspace*{1mm}\rule{0mm}{1mm}}\,}
\vspace{8pt}}
\def\cadre{$$\vcenter\bgroup\advance\hsize by -2em\noindent
             \refstepcounter{equation}(\theequation)~\ignorespaces}
\makeatletter
\def\endcadre{\egroup\eqno$$\global\@ignoretrue}
\def\ncadre{$$\vcenter\bgroup\advance\hsize by -2em\noindent
             \ignorespaces}
\makeatletter
\def\endncadre{\egroup\eqno$$\global\@ignoretrue}

\newcommand{\comment}[1]{}
\newcommand{\mybreak} {\par\vspace{2mm}\noindent}

\newcommand{\mv}[1] {\mathsf{#1}}

\makeatletter
\def\imod#1{\allowbreak\mkern10mu({\operator@font mod}\,\,#1)}
\makeatother

\newcommand{\C} {\mathcal{C}}

\newcommand{\F} {\mathcal{F}}
\newcommand{\HH} {\mathcal{H}}
\newcommand{\B} {\mathcal{B}}
\newcommand{\LL} {\mathcal{L}}
\newcommand{\K} {\mathcal{K}}

\newcommand{\NN} {\mathcal{N}}

\newcommand{\maxbic} {\B}

\newcounter{progcount}
\newcounter{linecount}[progcount]

\newcommand{\N}{\refstepcounter{linecount}\thelinecount. \>}
\newcommand{\NL}[1]{\refstepcounter{linecount}\thelinecount. \label{#1}\>}

\newenvironment{prog}[1]{
    \refstepcounter{progcount}\label{#1}
    \par\vspace{0.5ex}\noindent\hspace{1ex}
    \begin{minipage}{\linewidth}
    \small
    \begin{tabbing}
    =spa\=spa\=spa\=spa\=spa\=spa\=spa\=spa\=spa\=spa\=spa\=spa\=\kill
}%
{
    \end{tabbing}
    \end{minipage}\\[0.5ex]
}

\newcommand{\key}[1]{\textbf{#1~}}\ignorespaces

\begin{document}
\pagestyle{plain}

\title{On the Galois Lattice of Bipartite Distance Hereditary Graphs}

%
%

\author{Nicola Apollonio\footnote{Istituto per le Applicazioni del
Calcolo, M. Picone, v. dei Taurini 19, 00185 Roma, Italy.
\texttt{nicola.apollonio@cnr.it}} \and {Massimiliano
Caramia\footnote{Dipartimento di Ingegneria dell'Impresa,
Universit\`a di Roma ``Tor Vergata'', v. del Politecnico 1, 00133
Roma, Italy. \texttt{caramia@disp.uniroma2.it}}}
\and
{Paolo Giulio Franciosa\footnote{Dipartimento di Scienze Statistiche,
Sapienza Universit\`a di Roma, p.le Aldo Moro 5, 00185
Roma, Italy. \texttt{paolo.franciosa@uniroma1.it}. The author was partially supported by the Italian Ministry of Education,
University, and Research (MIUR) under PRIN 2012C4E3KT national
research project ``AMANDA -- Algorithmics for MAssive and Networked
DAta''}}
}

\date{}

\maketitle

\begin{abstract} We give a complete characterization of bipartite graphs having tree-like Galois lattices. We prove that the poset obtained by deleting
bottom and top elements from the Galois lattice of a bipartite
graph is tree-like if and only if the graph is a Bipartite
Distance Hereditary graph. By relying on the interplay between
bipartite distance hereditary graphs and series-parallel graphs,
we show that the lattice can be realized as the containment
relation among directed paths in an arborescence.
Moreover, a compact encoding of Bipartite Distance Hereditary graphs is proposed, that allows optimal time computation of neighborhood intersections and maximal bicliques.\\

\end{abstract}

\section{Introduction}\label{sec:intro}
Galois lattices are a well established topic in applied lattice theory. Their importance is widely recognized \cite{gw}, and its applications span across theoretical computer science and discrete mathematics as well as artificial intelligence, data mining and data-base theory. There is a growing interest on the interplay between finite Galois lattices and other discrete structures in combinatorics and computer science, and new relationships have been (and are to be) discovered between graphs and the related Galois lattices. This papers follows this stream and characterizes a class of bipartite graphs by the Galois lattice of their maximal cliques.
\mybreak
\emph{Distance Hereditary} graphs are graphs with the
\emph{isometric property}, i.e., the distance function of a
distance hereditary graph is inherited by its connected induced
subgraphs. This important class of graphs was introduced and
thoroughly investigated by Howorka in \cite{how1,how2}. A
\emph{comparability graph} is the graph of the comparability
relation among elements of a poset. In \cite{codiste}, Cornelsen
and Di Stefano proved that by intersecting the class of Distance
Hereditary graphs with the class of comparability graphs one
obtains precisely the comparability graphs of tree-like posets,
i.e., those posets whose transitive reduction is a tree. Here we
investigate another relation between comparability graphs and
distance hereditary graphs: inspired on the one hand by the work of Amilhastre,
Vilarem and Janssen \cite{avg} and on the other hand by the work of Berry and Sigayret \cite{bs} and the work of Brucker and G\'ely \cite{bg}. In \cite{avg}, \emph{Galois lattices} of domino-free bipartite graphs are investigated.
In \cite{bs} it is shown that the Hasse diagram of the \emph{Galois lattice} of \emph{chordal bipartite} graphs is \emph{dismantable} \cite{rival}, while an analogous result is shown in \cite{bg}  for the \emph{clique lattice} of \emph{strongly chordal graphs}. Both \cite{bs} and \cite{bg} use a dismantlability property of these lattices proved in
\cite{rival}. Recall that a graph $G$ is \emph{strongly chordal} if and only if its \emph{vertex-clique graph}, namely, the incidence bipartite graph of the maximal cliques of $G$ over $V(G)$, is a \emph{bipartite chordal} graph and that a graph is \emph{bipartite chordal} if it does not contain an induced copy of a chordless cycle on more than four vertices---the reader is referred to Section~\ref{sec:prel} for undefined terms and notions.  
\mybreak
In this paper we study the transitive
reduction of the \emph{Galois lattice} of those bipartite graphs that are chordal (as in \cite{bs}) and domino-free (as in \cite{avg}). It follows by Theorem~\ref{thm:bdhchar} in Section~\ref{sec:prel} that these graphs are precisely the Bipartite Distance
Hereditary (BDH for shortness) graphs, namely, those distance hereditary graphs which are bipartite.
\mybreak
Essentially in the same way as chordal bipartite graphs are related to strongly chordal graphs, BDH graphs are related to the so called \emph{Ptolemaic graphs}. If $\mathbf{CH}$ denotes the class of chordal graphs, namely, those graphs that do not contain an induced copy of the chordless cycle on more than three vertices, and if $\mathbf{DH}$ is the class of distance hereditary graphs, then the class $\mathbf{Pt}$ of Ptolemaic graphs is the intersection between $\mathbf{CH}$ and $\mathbf{DH}$. Actually, by the results of \cite{peledwu}, $\mathbf{Pt}$ is the intersection between $\mathbf{SC}$ and $\mathbf{DH}$, where $\mathbf{SC}$ is the class of strongly chordal graphs. Let $\LL(G)$ denote the \emph{Galois lattice} of a bipartite graph $G$ and let $\C(H)$ denote the \emph{clique lattice} of a graph $H$. As shown by Wu (as credited in \cite{peledwu}), if $G$ is Ptolemaic, then the vertex-clique graph of $G$ is a BDH graph. Hence there is a map $\lambda: \mathbf{Pt}\rightarrow \mathbf{BDH}$ and it is not difficult to see that $\LL(\lambda G)\cong \C(G)$, where $\cong$ is lattice isomorphism. In a sense, as we show in Section \ref{sec:indirect}, the converse statement holds as well, namely, there is a mapping $\mu$ that takes a a BDH graph $G$ into a Ptolemaic graph $\mu G$ so that $\C(\mu G)\hookrightarrow \LL(G)$ in such a way that $\LL(G-I)\cong \C(\mu G)$ for a certain set $I$ of \emph{join-irreducible} (or \emph{meet-irreducible}) elements of $\LL(G)$, where $\hookrightarrow$ denotes \emph{order embedding}. In other words the following diagram applies (and commutes):

\begin{equation}\label{eq:diagram2}
\xymatrix@1{
\mathbf{Pt}\ar@{<->}[d]_{\C(\cdot)}\ar@<1ex>[rrrr]^{\lambda} & &&& \mathbf{BDH}\ar@<1ex>[llll]^{\mu}\ar@{<->}[d]^{\LL(\cdot)} \\
\mathbb{T}_{\K}\ar@/_/[rrrr]|{\Psi_\mu} & & && \mathbb{T}_{\B}\ar@/_/[llll]|{\Phi_\lambda} }
\end{equation}
where $\mathbb{T}_\K$ and $\mathbb{T}_\B$ are the classes of tree-shaped clique lattices and Galois lattices, respectively, $\Phi_\lambda$ is lattice isomorphism induced by $\lambda$, and $\Psi_\mu$ is an order embedding induced by $\mu$.
\subsection{Our result}
Let us recall what is the \emph{Galois
lattice} of a bipartite graph $G$. Let $G$ have color classes $X$
and $Y$. A \emph{biclique} of $G$ is a set $B\subseteq V(G)$ which
induces a complete bipartite graph. Let $\B(G)$ be the
set of the (inclusionwise) maximal bicliques of $G$ and for $B\in
\B(G)$ let $X(B)=B\cap X$ and $Y(B)=B\cap Y$. $X(B)$ and
$Y(B)$ are called the \emph{shores} of $B$. Throughout the rest of the paper we assume that $G$ does not contain universal vertices, 
where a \emph{universal vertex} in a bipartite graph is a vertex that is adjacent to all vertices in the opposite color class. This assumption, while does not cause loss of generality, leads to simpler statements and proofs.
Following \cite{avg},
we endow $\B(G)$ by a partial order $\preceq$ defined by
$$B\preceq B' \Leftrightarrow X(B)\subseteq X(B').$$
Equivalently, the same partial order can be defined as
$$B\preceq B' \Leftrightarrow Y(B)\supseteq Y(B')$$
since $X(B)\subseteq X(B') \Leftrightarrow Y(B)\supseteq Y(B')$.
If we extend $\B(G)$ by
adding two dummy elements $\bot$ and $\top$ acting as bottom and
top element respectively, the poset
$\mathcal{L}(G)=(\maxbic(G)\cup\{\bot,\top\},\preceq)$ is a lattice
known as the Galois lattice of $G$. The two dummy elements are
respectively defined by
$$X(\bot)=\emptyset,\, Y(\bot)=Y \,\,\mbox{and}\,\, X(\top)=X,\, Y(\top)=\emptyset.$$

In this paper we prove that the shape of
$\mathcal{L}(G)$ can be used to characterize BDH graphs. More
precisely, we show the following.
\begin{theorem}\label{thm:main} Let $G$ be a connected bipartite
graph and let $\mathbf{H}(G)$ be the transitive reduction of
$(\mathcal{B}(G),\preceq)$. Then $\mathbf{H}(G)$ is a tree if and
only if $G$ is a BDH graph.
\end{theorem}
Otherwise stated: after deleting $\bot$ and $\top$,  $\mathcal{L}(G)$ is a
tree-like poset. This is a very strong property: for instance, it
allows efficient enumeration of linear extensions \cite{atk}.
The question of studying bipartite graphs (binary relations) whose Galois lattice is tree-like (arborescence-like in a sense) was raised first in \cite{bdov}. Here we completely solve the problem from a graph-theoretical view-point. We also give a direct proof that $(\mathcal{B}(G),\preceq)$ has dimension at most 3, though this can be derived by known properties of planar posets~\cite{trmoo3}.
\mybreak
Although, as we show in Section \ref{sec:indirect}, Theorem \ref{thm:main} can be deduced with some extra work from other known results on graphs and hypergraphs (by taking the longest dipath in Diagram \ref{eq:diagram2}) the proof we present here is direct and self-contained. 
\mybreak
The corpus of theoretical and algorithmic machinery on Galois lattices would certainly allow a pure lattice theoretical development of the present paper. Nevertheless, we prefer to present the result as much combinatorially as possible. For instance, a bipartite graph is the counterpart of the lattice theoretical notion of \emph{context induced by a binary relation} and maximal bicliques are the graph theoretical counterparts of \emph{concepts} in Galois lattices \cite{gw}. In this respect we prefer to look at the covering diagram of a poset as to a directed acyclic graph and to think of ``crowns'' as cycles in an undirected graph. This choice allows us from the one hand to make the paper very self-contained and on the other hand to elicit the purely combinatorial arguments behind our proofs. As an example consider Lemma \ref{lem:Htree1}. That easy result is also a rather straightforward consequence of the fact that the Galois lattice of a sub-context of a given context can be order-embedded canonically into the Galois lattice of the context so that the covering diagram of the former is the covering diagram of an induced subposet (not necessarily a sublattice) of the latter. The proof we give is nothing but than a specialization of the general arguments used in Proposition 32 in \cite{gw}, but has the advantage of avoiding more sophisticated notion that would remain otherwise unused throughout the rest of the paper. Also notions such as \emph{join-irreducibility} and  \emph{meet-irreducibility} will be briefly recalled and referred to when they come out and mostly from a diagrammatic approach, because what matters for our purposes is their graphical counterpart. 

We warn the reader that our results do not use assumptions about the context usually introduced under the clause ``without loss of generality'', such as that the context is reduced or clarified. While this fact is not a gain of generality from a theoretical view-point it is surely a gain in algorithmic robustness and computational complexity.

\subsection{Relations with other classes of graphs}
BDH graphs have very strong structural properties and some of them are
highlighted in the characterization by Bandelt and Mulder \cite{bm} recalled in Theorem \ref{thm:bdhchar} (see also \cite{audamo,bm} and the
monograph \cite{bralespi}). Moreover, BDH graphs are related to other very well known classes of graphs: Ptolemaic graphs and series parallel graphs. The relation between Ptolemaic graphs and BDH graphs has been partially introduced above and will be further pursued in Section \ref{sec:indirect}. Let us discuss here the relation with series parallel graphs.

Ellis-Monaghan and Sarmiento \cite{mosa} showed that the
class of bipartite distance hereditary graphs is a nice
nontrivial class of polynomially computable instances for the
\emph{vertex-nullity interlace polynomial} introduced by Arratia,
Bollob\'{a}s and Sorkin in \cite{abs}, under the name of
\emph{interlace polynomial}. The former authors achieved
their result by exploiting a strong topological relationship between BDH
graphs and series-parallel graphs. In this paper we prove another deep relation between the two classes, namely BDH graphs are \emph{fundamental graphs} of series parallel graphs (see Section \ref{sec:encoding}). This result leads to an implicit representation of the Galois lattice of a BDH graph as a collection of paths in an arborescence. We further discuss this representation in Section \ref{sec:encoding}, where we exploit it to
show how the Galois lattice of a BDH graph, and the BDH graph itself, can be efficiently
encoded. The encoding of the BDH graph  requires $O(n)$ space in the worst case, $n$ being the order of the graph, still allowing the retrieval of the neighborhood of any vertex in time linear in the size of the neighborhood. Moreover, intersections of neighborhoods can be listed in optimal linear time in the size of the intersection, in the worst case. 

\paragraph{Organization} The rest of the paper goes as follows. We first give some preliminary notions in Section~\ref{sec:prel}.
In Section~\ref{sec:closprop} we prove a useful property of BDH
graphs. Such a property is then exploited in Section \ref{sec:char} to characterize BDH graphs through their Galois lattice---the property is intimately related with Fagin's results \cite{fagin} (as it is discussed in Section \ref{sec:indirect})--. In Section
\ref{sec:encoding} we show how the Galois lattice of any BDH graph can be
encoded as the containment relation among dipaths in an
arborescence. In Section \ref{sec:algorithms} we draw some algorithmic consequences of the encoding and, finally, in Section \ref{sec:indirect} we give another proof of Theorem \ref{thm:main} relying, via Diagram \ref{eq:diagram2}, on known results on \emph{Ptolemaic graphs}, \emph{$\gamma$-acyclic hypergraphs} and \emph{clique lattices of graphs}. 

\section{Preliminaries}\label{sec:prel}
Let $V$ be a finite set. By a \emph{hypergraph} on $V$ (with a little abuse of language) we simply mean a family $\HH$ of subsets of a given ground set $V$. Notice that $\HH$ can contain repeated members.
If $\HH$ is a hypergraph, then $\Gamma(\HH)$ is the \emph{bipartite incidence graph of} $\HH$ over $V$, that is, the bipartite graph with color class $V$ and $\HH$ where there is an edge between $v\in V$ and $F\in \HH$ if and only if $v\in F$.
\mybreak
If $G$ is a graph, then $V(G)$ denotes its
vertex-set if $G$ is undirected, while it denotes its node-set if
$G$ is directed. Similarly, $E(G)$ denotes the edge-set of $G$ if
$G$ is undirected, and the arc-set of $G$ if $G$ is directed. The
distance between two vertices $u$ and $v$ of an undirected graph
$G$, denoted by $d_{G}(u,v)$, equals the minimum
length of a path having $u$ and $v$ as end-vertices, or is $\infty$ if no such path exists. For a graph
$G$ and a vertex $v\in V(G)$, $N_G(v)$ (or simply $N(v)$ when $G$
is understood) is the set of vertices adjacent to $v$ in $G$. The \emph{degree} of $v$ is the number of vertices in $N_G(v)$. The graph induced by $V(G)-\{v\}$ is denoted by $G-v$. Let
$G$ be a directed graph and $v$ be a node of $G$.
We split the neighborhood of $v$ into
$N^-(v)=\{u\in V(G) \ |\ (u,v)\in E(G)\}$ and $N^+(v)=\{w\in V(G)
\ |\ (v,w)\in E(G)\}$. 
The
\emph{outdegree} of $v$ in $G$, denoted by $\deg^+_G(v)$, is the number $|N^+(v)|$ of arcs leaving $v$.
Analogously, the \emph{indegree} of $v$ in $G$, $\deg^-_G(v) = |N^-(v)|$,
is the number of arcs entering $v$.
A node in $G$ is a \emph{source} if its
indegree in $G$ is zero, a \emph{sink} if its
outdegree in $G$ is zero, or a \emph{flow-node} if it is
neither a source nor a sink. A \emph{dipath} $P$ of $G$ is a path
of $G$ with exactly one source in $P$ and exactly one sink in $P$.
A \emph{circuit} $C$ in $G$ is a cycle in $G$ with no source and
no sink in $C$.

The chordless cycle on
$n\geq 4$ vertices is denoted by $C_n$, and a \emph{hole} in a
bipartite graph is an induced subgraph isomorphic to $C_n$ for
some $n\geq 6$. A \emph{domino} is a subgraph isomorphic to the
graph obtained from $C_6$ by joining two antipodal vertices by a
chord (see Figure~\ref{fi:h}). A
$(l,k)$-\emph{chordal graph} is a graph such that every cycle of
length at least $l$ has at least $k$ chords. Bipartite
$(6,1)$-chordal graphs are simply called \emph{chordal bipartite}.
A twin of a vertex $v$ in a graph is a vertex with the same
neighbors as $v$.
\begin{theorem}[Bandelt and Mulder \cite{bm}, Corollaries 3 and 4]\label{thm:bdhchar} The
following statements are equivalent for a bipartite graph $G$:
\begin{enumerate}[(i)]
\item $G$ is a BDH graph;
\item\label{it:bandeltmount} $G$ is constructed from a single vertex by a sequence of adding
pending vertices and twins of existing vertices;
\item $G$ contains neither holes nor induced dominoes;
\item  $G$ is a bipartite $(6,2)$-chordal graph.
\end{enumerate}
\end{theorem}
If $G,H_1,H_2\ldots, H_n$ are graphs, we say that $G$ is
$H_1,\ldots, H_n$-\emph{free} if $G$ contains no induced copy of
$H_i$, $i=1,\ldots, n$. Funny enough, after Theorem
\ref{thm:bdhchar}, one can say that a graph is BDH if and only if
it is DH-free: just solve the latter acronym as Domino
Hole. 

In a poset $(X,\leq)$ an element $y$ covers an element $x$ if
$x\leq y$ and $x\leq z$ $\Rightarrow$ $y\leq z$. If $x,\,y$ are
incomparable we write $x\parallel y$. The \emph{least} or \emph{bottom element} of a poset $(X,\leq)$ is the unique element $x\in X$ such that $x\leq x'$ for every $x'\in X$. This element is usually denoted by $\bot$. The \emph{greatest} or \emph{top element} of $(X,\leq)$, usually denoted by $\top$, is defined dually. The transitive reduction of a poset $(X,\leq)$ is the directed acyclic graph on
$X$ where there is an arc leaving $x$ and entering $y$ if and
only if $y$ covers $x$. The meet and the join operators in a
lattice are denoted as customary by $\wedge$ and $\vee$,
respectively. 

An element $x$ in a poset $(X,\leq)$ is \emph{meet-irreducible} (resp., \emph{join-irreducible}) if $x = y \wedge z$ (resp., $x = y \vee z$) implies $x=y$ or $ x=z$.
Let $(X_1,\leq_1)$ and $(X_2,\leq_2)$ be two posets. An order embedding of $(X_1,\leq_1)$ into $(X_2,\leq_2)$ is a map $f: X_1\rightarrow X_2$ satisfying the following condition
$$x\leq_1 y\Longleftrightarrow f(x)\leq_2f(y).$$  
An \emph{order isomorphism} is a bijective order embedding.
\mybreak
For a bipartite graph $G$ let
$\mathcal{L}^\circ(G)=(\maxbic(G),\preceq)$ and recall that $\mathcal{L}(G)$ denotes $(\maxbic(G) \cup \{\top,\bot\},\preceq)$. Thus $\mathbf{H}(G)$ is the transitive
reduction of $\mathcal{L}^\circ(G)$.
Throughout the rest of the paper we represent a biclique $B$ of a bipartite graph $G$ by the ordered pair of its shores, i.e, we write $B=(U,W)$ to mean that $U=X(B)$, $W=Y(B)$ and that $X(B)\cup Y(B)$ induces a complete bipartite subgraph of $G$. Moreover, with some abuse of notation, if $v\in V(G))$, then we write $v\in B$ to mean that $v\in X(B)\cup Y(B)$ and, analogously, we write $B-v$ for the biclique induced by $(X(B)\cup Y(B)-\{v\}$. A biclique $B$ \emph{dominates} a biclique $B'$ if $X(B')\subseteq
X(B)$ and $Y(B')\subseteq Y(B)$.
\begin{figure}
    \begin{center}
             \includegraphics[width=8.5cm]{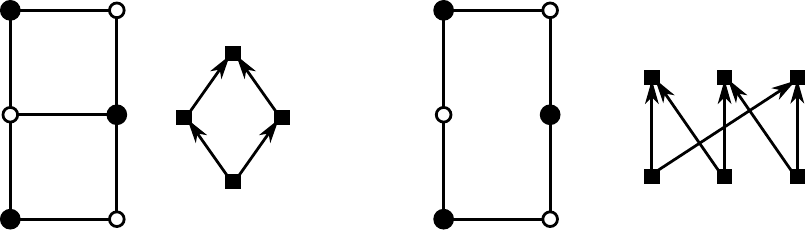}
    \end{center}
    \caption{Domino and $C_6$ and the corresponding Galois lattices.}
    \protect\label{fi:h}
\noindent\hrulefill%
\end{figure}
As an example, let $G$ be either the domino or the $C_6$ (see Figure~\ref{fi:h}). If $G$ is the domino, then $\maxbic(G)$ contains four
members: the vertex-sets of the two stars centered at vertices of
degree three and the vertex-sets of two squares; $\mathbf{H}(G)$
is thus a directed square with one source and one sink; if $G$
is the $C_6$ then the members of $\maxbic(G)$ are the
vertex-sets of the subpaths of $G$ of length 2; therefore,
$\mathbf{H}(G)$ is a directed $C_6$ with three sources and three
sinks.
\begin{remark}\label{rem:usef2} For $B,\,B'\in
\maxbic(G)$ one has $B\parallel B'$ if and only if
$\{X(B),\,X(B')\}$ and $\{Y(B),\,Y(B')\}$ both have inclusionwise
incomparable members. Indeed, if $X(B)\subseteq X(B')$, say, then
$X(B)\cup(Y(B)\cup Y(B'))$ is a biclique of $G$ dominating $B$.
\end{remark}
\begin{remark}\label{rem:usef3} Galois lattices are self-dual in the following sense: if $\mathcal{L}(G)$
is the Galois lattice of $G$ then $\mathcal{L}^*(G)$ (the lattice
dual of $\mathcal{L}(G)$) is the Galois lattice of $G$ with color
classes interchanged. We often use this fact later in the
following way: if we prove a property of the lattice for the
$X$-shores of maximal bicliques, then the same property holds by
duality for the $Y$-shores.
\end{remark}
If $X_0\subseteq X$, then there is a biclique $B_0\in \mathcal{L}^\circ(G)$ such that $X(B_0)=X_0$ if and only if $X_0=\bigcap_{y\in Y_0} N(y)$ for some $Y_0\subseteq Y$. 
Analogously if $Y_0\subseteq Y$, then there is a biclique $B_0\in \mathcal{L}^\circ$ such that $Y(B_0)=Y_0$ if and only if $Y_0=\bigcap_{x\in X_0} N(x)$ for some $X_0\subseteq X$.
Using these facts one has that the projections  $(X_0,Y_0)\mapsto X_0$ and $(X_0,Y_0)\mapsto Y_0$ are actually order isomorphism between $\mathcal{L}^\circ(G)$ and $\{X(B) \ |\ B\in \mathcal{L}^\circ(G)\}$ and $\{Y(B) \ |\ B\in \mathcal{L}^\circ(G)\}$. Hence 
\begin{equation}\label{eq:shoresisosx}
\mathcal{L}^\circ (G)\cong \Big(\big\{X(B) \ |\ B\in \mathcal{L}^\circ(G)\big\}, \subseteq\Big)
\end{equation}
and
\begin{equation}\label{eq:shoresisodx}
\mathcal{L}^\circ (G) \cong \Big(\big\{Y(B) \ |\ B\in \mathcal{L}^\circ(G)\big\}, \supseteq\Big)
\end{equation}
(see also \cite{gw}).


\section{A Closure Property}\label{sec:closprop}
Before proceeding toward the proof of Theorem \ref{thm:main}, we
discuss separately a sort of ``convexity property'' for the
neighborhood of the vertices of a BDH graph. Such a property is stated in Theorem \ref{thm:convex} and it is needed to prove the necessity in Theorem \ref{thm:main}, besides, we deem it interesting on its
own. In Section~\ref{sec:indirect} we show that Theorem \ref{thm:convex} is equivalent to one of Fagin's results~\cite{fagin}, namely, to the first implication in Theorem~\ref{thm:fagin}. Let $G$ be a connected BDH graph. For
$v,\, v'\in V(G)$, let $G\star \{v,\,v'\}$ be the graph defined as
follows:
\begin{itemize}
\item[--] if $v$ and $v'$ are in different color classes, then $G\star
\{v,\,v'\}$ is $G$;
\item[--] if $v$ and $v'$ are in the same color class, then $G\star
\{v,\,v'\}$ is obtained from $G$ by adding a new vertex
$\widehat{vv'}$ to the color class of $v$ and $v'$. Vertex $\widehat{vv'}$ is adjacent to every vertex in $N(v)\cap N(v')$.
\end{itemize}
\begin{theorem}\label{thm:convex} Let $G$ be a BDH graph and let $v,\,v'\in V(G)$. Then $G\star \{v,\,v'\}$ is a BDH
graph, that is the class of BDH graphs is closed under $\star$.
\end{theorem}
\begin{proof}
Without loss of generality, let $v,\, v'\in X$ and to simplify the
notation let $v_0=\widehat{vv'}$ and $G_0=G\star \{v,\,v'\}$. By
contradiction suppose that $G_0$ contains induced dominoes or holes
while $G$ does not. Suppose first that $D$ is an induced domino in
$G_0$. Clearly $v_0\in V(D)$. Let $d$ be the number of neighbors of
$v_0$ in $D$. Then $d=2$. To see this observe that if $d=3$ then
neither $v$ nor $v'$ can belong to $D$, because both $v$ and $v'$ should be adjacent to three vertices in $D$, contradicting the assumption that $D$ is domino-free.
Since $v\not\in V(D)$ it
follows that $V(D)\Delta \{v_0,v\}$ induces a domino in $G$
contradicting that $G$ is domino-free. Thus $d=2$. Let $y_1$,
$y_2$ and $y_3$ be the vertices of $D$ contained in $Y$, with,
say, $y_1$ and $y_2$ adjacent to $v_0$ and $y_3\not\in N(v)\cap
N(v')$. At least one among $v$ and $v'$ does not belong to $V(D)$,
otherwise $y_3$ would belong to $N(v)\cap N(v')$. Suppose first
that exactly one of them is in $V(D)$ and let, without loss of
generality, $v\in V(D)$ and $v'\not\in V(D)$. Then, $y_3\in N(v)$
and $y_3\not\in N(v')$, because $y_3\not\in N(v)\cap N(v')$.
Therefore, $V(D)\Delta\{v_0,v'\}$ induces a domino in $G$, still a
contradiction. We must conclude that $v$ and $v'$ do not belong to
$D$. Now the graphs $D_1$ and $D_2$ induced in $G_0$ by
$V(D)\Delta\{v_0,v\}$ and $V(D)\Delta\{v_0,v'\}$, respectively,
are both induced subgraphs of $G$. Since $G$ is domino-free
neither of them is a domino. Therefore, $y_3$ must be adjacent to
$v$ in $D_1$ and to $v'$ in $D_2$ implying that $y_3\in N(v)\cap
N(v')$. The latter contradiction proves that $G_0$ is domino-free.

Suppose now that $G_0$ contains a hole $H$. The vertex-set of such
a hole must contain $v_0$ and contains neither $v$ nor $v'$
(otherwise $H$ would be chorded). Now the graph induced in $G_0$ by
$V(D)\Delta\{v_0,v\}$ is a subgraph of $G$ and therefore, by
Statement (iv) of Theorem \ref{thm:bdhchar}, $H$ possesses at
least two chords and such two chords must be incident to $v$.
Choose one of the possible orientations of $H$ and let $y,\,y'\in
V(H)\cap G$ the end-vertices of the chords incident to $v$ met as
first and as second, respectively, while traveling on $H$
starting from $v$ in the chosen orientation. Let $R$ be the set of
vertices met after $y'$ and before coming back to $v$. Now
$V(H)\setminus R$ induces a cycle with exactly one chord. A
contradiction.
\end{proof}
Let
$\mathcal{X}_G=\left(X(B) \ |\ B\in \maxbic(G)\right)$ and
$\mathcal{Y}_G=\left(Y(B) \ |\ B\in \maxbic(G)\right)$.

\begin{corollary}\label{lem:tb} If $G$ is a BDH graph then
so are the graphs $\Gamma(\mathcal{X}_G)$ and $\Gamma(\mathcal{Y}_G)$.
\end{corollary}
\begin{proof} By duality it suffices to prove the lemma only for
$\mathcal{Y}_G$. One has $W\in \mathcal{Y}$ if and only if $(U,W)\in \maxbic(G)$ for some $U\subseteq X$ and $W=\bigcap_{u\in
U}N_G(u)$. Therefore, $\mathcal{Y}_G$ is a subfamily of the family
$$\mathcal{C}=\left(\bigcap_{u\in U}N_G(u) \ |\ U\subseteq X\right)$$
and $\Gamma(\mathcal{Y}_G)$ is an induced subgraph of
$\Gamma(\mathcal{C})$. Observe that
$\Gamma(\mathcal{C})\cong \Gamma(\{N_{\tilde{G}}(x) \ |\ x\in
\tilde{X}\})$ for a certain graph $\tilde{G}$ with color classes
$\tilde{X}$ and $Y$ arising from $G$ by a repeated application of
operation $\star$. Such an operation preserves the property of
being a BDH graph. Thus $\Gamma(\mathcal{C})$ (and hence
$\Gamma(\mathcal{Y}_G)$) is BDH.
\end{proof}

\section{Characterizing Chordal, Domino-Free Bipartite Graphs by their Galois
lattices}\label{sec:char}
In this section we prove Theorem
\ref{thm:main}. The proof of the \emph{if part} is given in
Section \ref{sec:if} while the \emph{only if part} is proved in
Section \ref{sec:onlyif}.

\subsection{Proof of the \emph{if part}}\label{sec:if}
Let us exploit now the structure of BDH graphs to prove the
\emph{if part} of Theorem \ref{thm:main}. We remark that the next two
results apply to the more general class of domino-free bipartite graphs.
\begin{lemma}\label{lem:galois2} If $G$ is a domino-free bipartite graph then
for any $B^1,\,B^2\in\maxbic(G)$ such that $B^1 \parallel B^2$ one has
$$\bot\not= B^1\wedge B^2\ \Rightarrow\ B^1\vee B^2=\top\ \ \ \ \mathrm{and}
\ \ \ \ B^1\vee B^2\not=\top\ \Rightarrow\ B^1\wedge B^2=\bot.$$
\end{lemma}
\begin{proof} By the sake of contradiction assume that
$$\bot\not= B^1\wedge B^2\prec B^1\vee B^2\not=\top$$
for some $B^1,\,B^2\in\maxbic(G)$ and let $B^0=B^1\wedge B^2$
and $B^3=B^1\vee B^2$. A vertex $v\in V(B^1)\cup V(B^2)$ is called
\emph{heavy} if it is a universal vertex in the subgraph induced by $B^1 \cup B^2$. The maximality of $B^1$ and $B^2$
implies that there is no heavy vertex in $B^1 \Delta B^2$. Now
$X(B^0)$ and $Y(B^3)$ are both nonempty because $B^0\not=\bot$ and
$B^3\not=\top$. Thus, we can pick $x_0\in X(B^0)$ and $y_3\in
Y(B^3)$. Hence, $x_0y_3\in E(G)$. Now pick $x_1\in X(B^1)\setminus
X(B^2)$ and $y_2\in Y(B^2)\setminus N(x_1)$; the latter vertex
exists because $x_1$ is not heavy. Similarly, pick $x_2\in
X(B^2)\setminus X(B^1)$ and $y_1\in Y(B^2)\setminus N(x_2)$; the
latter vertex exists because $y_1$ is not heavy. The subgraph
induced by $\{x_0,x_1,x_2,y_1,y_2,y_3\}$ is a domino, contradicting
the hypothesis.
\end{proof}
\begin{figure}
    \begin{center}
             \includegraphics[width=8.5cm]{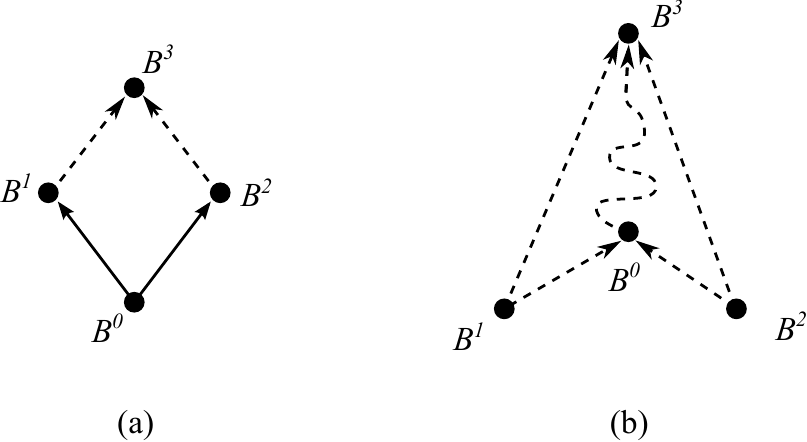}
    \end{center}
    \caption{Proof of Lemma \ref{lem:noflow}; dotted lines represent dipaths and hence chains.}
\protect\label{fi:galois_forb}
\noindent\hrulefill%
\end{figure}
\begin{lemma}\label{lem:noflow} Let $G$ be a domino-free bipartite
graph and $\mathbf{H}(G)$ be the transitive reduction of
$\mathcal{L}^\circ(G)$. Then, any cycle of $\mathbf{H}(G)$ that does not contain $\bot$ or $\top$ has at
least six non-flow-nodes.
\end{lemma}
\begin{proof}
First notice that any cycle of any directed graph has an even
number of non-flow nodes because it has as many sources as sinks.
Moreover, if the graph is acyclic then such a number is positive.
Transitive reductions are triangle-free. Therefore, if
$\mathbf{C}$ is a cycle in $\mathbf{H}(G)$, then $\mathbf{C}$ has
at least four nodes and the number of its non-flow-nodes is even
and positive. Let us prove first that $\mathbf{C}$ cannot have
exactly two non-flow-nodes. Referring to 
Figure~\ref{fi:galois_forb} (a), let $B^0$ and $B^3$ be the source and the sink of
$\mathbf{C}$, respectively, and let $B^1$ and $B^2$ be the
neighbors of $B^0$ in $\mathbf{C}$. Clearly $B^0=B^1\wedge
B^2\not=\bot$. The existence of dipaths from $B^1$ to $B^3$ and
$B^2$ to $B^3$ implies that $B^1\vee B^2\prec B^3\not=\top$,
contradicting Lemma \ref{lem:galois2}. So $\mathbf{C}$ must have
at least four non-flow-nodes and it is of the form shown in Figure~\ref{fi:galois_forb} (b), where $B^1$ and $B^3$ are sources, while $B^0$ and $B^2$ are sinks. As above, the existence of the dipaths from
$B^i$ to $B^0$ and from $B^i$ to $B^3$, $i=1,2$, implies that
$B^1\vee B^2\prec B^0$ and $B^1\vee B^2\prec B^3$. So, possibly by
replacing $B^0$ by $B^1\vee B^2$ and redefining $\mathbf{C}$, we
may assume that $B^1\vee B^2=B^0$. By definition, $B^0$ is the only least
upper bound. Thus $B^0\prec B^3$ and there exists a dipath
$\mathbf{P}$ from $B^0$ to $B^3$. Now $\mathbf{C}\cup \mathbf{P}$
contains two cycles with exactly two non-flow-nodes contradicting
the preceding part of the proof. We conclude that $\mathbf{C}$ has
at least six non-flow-nodes.
\end{proof}

We are now ready the prove the \emph{if part} of Theorem \ref{thm:main}.\\

\noindent\textbf{Proof of the \emph{if part} of Theorem
\ref{thm:main}.} We assume $\bot$ and $\top$ have been deleted from $\mathbf{H}(G)$. Since $\mathbf{H}(G)$ is connected we have only
to show that it does not contain cycles. Suppose by contradiction
that $\mathbf{H}(G)$ contains some cycle, and let
$\mathbf{C}$ be a cycle having the least possible number of non-flow-nodes. Let
$2t$, $t\in \mathbb{N}$, be such a number. As $G$ is a BDH graph it
is domino-free. Therefore, by Lemma
\ref{lem:noflow}, $t\geq 3$. Let
$B^1,\ldots,B^{2t-1}$ and $B^2,\ldots,B^{2t}$ be the sources and
the sinks of $\mathbf{C}$, respectively, as they are met
traversing the cycle in a chosen direction. By definition of
transitive reduction one has
$$\emptyset\not=X(B^1)\subseteq X(B^2)\cap
X(B^{2t})$$
and
$$\emptyset\not=X(B^{2i+1})\subseteq X(B^{2i})\cap
X(B^{2(i+1)}),\, i=1,\ldots, t-1.$$ Moreover, for
$i\in\{0,\ldots,t-1\}$ and $j\in\{1,\ldots,t\}$ such that
$|i-j|\not\in\{0,1,t\}$ one has $X(B^{2i+1})\cap
X(B^{2j})=\emptyset$. Otherwise $X(B^{2i+1})\wedge X(B^{2j})\in
V(\mathbf{H}(G))$ and one of the two subpaths of $\mathbf{C}$
connecting $X(B^{2i+1})$ and $X(B^{2j})$ along with the two paths of
$\mathbf{H}(G)$ connecting $X(B^{2i+1})\wedge X(B^{2j})$ to
$X(B^{2i+1})$ and $X(B^{2j})$ respectively, would define a cycle
$\mathbf{C}'$ of $\mathbf{H}(G)$ with fewer non-flow-nodes than
$\mathbf{C}$. Now for $i=0,\ldots, t-1$, pick $x_{2i+1}\in
X(B^{2i+1})$ and let $U=\{x_1,x_3\ldots,x_{2t-1}\}$ and
$\mathcal{U}=\{X(B^{2i}) \ |\ i=1,\ldots t\}$. Thus $U\cup
\mathcal{U}$ induces a hole in $\Gamma(\mathcal{X}_G)$,
contradicting Corollary \ref{lem:tb}.
\subsection{Proof of the \emph{only if part}}\label{sec:onlyif} To complete the proof of Theorem
\ref{thm:main} we need some more properties of $\mathbf{H}(G)$. 

\begin{lemma}\label{lem:galois2bis} Let $G$ be a BDH
graph with at least three vertices. Then $x\in X$ is a cut-vertex
of $G$ if and only if $(\{x\},N(x))\in \maxbic(G)$.
Analogously, $y\in Y$ is a cut-vertex of $G$ if and only if
$(N(y),\{y\})\in \maxbic(G)$.
\end{lemma}
\begin{proof}
By duality it suffices to prove the lemma only for $x\in X$.
Suppose that $x$ is a cut-vertex of $G$. There is no biclique
$B$ with $x\in X(B)$, $|X(B)| \geq 2$  and $N(x)\subseteq Y(B)$, otherwise the
removal of $x$ cannot disconnect $G$. Hence $(\{x\}, N(x)) \in \B(G)$.

Conversely, let $B=(\{x\}, 
N(x)) \in \B(G)$ and suppose, for the sake of
contradiction, that $x$ is not a cut-vertex of $G$. As the removal
of $x$ does not disconnect $G$, any two neighbors $y$ and $y'$ of
$x$ have a common neighbor $x'\in X$. This follows by the fact
that $G-x$ is connected and hence $d_{G-x}(y,y')=d_{G}(y,y')=2$,
because $G$ is distance hereditary. Therefore,
$(\{x,x'\},\{y,y'\})$ is a biclique and there exists a biclique
$B' \in \maxbic(G)$ such that $B \prec B'$ and both shores of
$B'$ have at least two elements. Among the bicliques of
$\maxbic(G)$ fulfilling these conditions choose one with
$|Y(B)|$ as large as possible. Let this biclique be $\tilde{B}$.
The maximality of $B$ implies that $Y(\tilde{B})$ is strictly
contained in $N(x)$. Let $y_1\in N(x)\setminus Y(\tilde{B})$. One
has $N(y_1)\cap X(\tilde{B})=\{x\}$ by the choice of
$Y(\tilde{B})$. Indeed, if there were $u\in (X(\tilde{B})\setminus
\{x\})\cap N(y_1)$ then $\{x,u\}\cup Y(\tilde{B})\cup \{y_1\}$
would be a biclique with $|Y(\tilde{B})\cup
\{y_1\}|=|Y(\tilde{B})|+1>|Y(\tilde{B})|$. As $x$ is not a
cut-vertex, $y_1$ is not a pending vertex. Hence
$|N(y_1)\setminus\{x\}|\geq 1$. The choice of $\tilde{B}$ also
implies that for each $u\in N(y_1)\setminus\{x\}$ we can find at
least one vertex of $Y(\tilde{B})$ which is not adjacent to $u$
(otherwise the maximality of $\tilde{B}$ would be contradicted).
On the other hand, for each $y\in Y(\tilde{B})$ there is some
$v\in N(y_1)\setminus\{x\}$ which is adjacent to $y$. This because
$G-x$ is connected and $d_{G-x}(y,y_1)=d_{G}(y,y_1)=2$. Therefore,
we can find $x_1\in N(y_1)\setminus\{x\}$ and $y_2,\,y_3\in
Y(\tilde{B})$ in such a way that $y_2$ is a neighbor of $x_1$
while $y_3$ is not. Finally, as $|X(\tilde{B})|\geq 2$, we can find
$x_2\in X(\tilde{B})\setminus \{x\}$. But now
$(\{x,x_1,x_2\},\{y_1,y_2,y_3\})$ induces a domino in $G$. A
contradiction which proves the lemma.
\end{proof}
Recall that in poset that has a bottom element $\bot$, an \emph{atom} is an element of the poset that covers $\bot$. Dually, if the poset has a top element $\top$, a \emph{co-atom} is an element which is covered by $\top$. After this terminology we can say that the cut vertices of $G$ are either atoms or co-atoms.
\mybreak
We now study the behavior of $\mathbf{H}(G-v)$ for $v\in V(G)$.
Let us begin
with an easy but useful property of $\textbf{H}(G)$ in the general
case. The next lemma proves that if the deletion
of a vertex $v$ from a maximal biclique of $G$ does not cause loss of
maximality in the biclique, then $\mathbf{H}(G-v)$ inherits from
$\mathbf{H}(G)$ as much adjacency as possible.
\begin{lemma}\label{lem:Htree1}
Let $G$ be a bipartite graph, $B_0\in \maxbic(G)$ and $v\in B_0$. If
$B_0-v\in\maxbic(G-v)$, then
\begin{itemize}
\item[--]there is an arc $(B-v,B_0-v)$ in $\mathbf{H}(G-v)$ for every $B\in\maxbic(G)$ such that $B-v\in \maxbic(G-v)$and $(B,K_0)$ is an arc of $\mathbf{H}(G)$;
\item[--]there is an arc $(B_0-v,B-v)$ in $\mathbf{H}(G-v)$ for every $B\in\maxbic(G)$ such that $B-v\in \maxbic(G-v)$and $(B_0,B)$ is an arc of $\mathbf{H}(G)$;
\end{itemize}
In other words,
$$\phi: \mathbf{H}(G-v)\ni B-v\longmapsto B\in \mathbf{H}(G)$$
embeds $\mathbf{H}(G-v)$ in $\mathbf{H}(G)$ as a sub-digraph.
\end{lemma}
\begin{proof}
By duality, it suffices to prove the
statement only when $v\in X(B)$. Let $B\in\maxbic(G)$ be such that
$B-v\in \maxbic(G-v)$. The thesis of
the lemma follows by the equivalences listed below:
$$
\left(B_0\parallel B\,\Leftrightarrow\, B_0-v\parallel
B-v\right);\,\, \left(B_0\prec B\,\Leftrightarrow\, B_0-v\prec
B-v\right);\,\, \left( B\prec B_0\,\Leftrightarrow\, B-v\prec
B_0-v\right).
$$
To prove these statements it suffices to recall Remark
\ref{rem:usef2} and to observe that since $v\in X$ one has
$Y(B_0-v)=Y(B_0)$ and $Y(B-v)=Y(B)$. Therefore, if $B_0-v$ and
$B-v$ are both maximal in $G-v$ then they must be in the same
relation in $\mathcal{L}(G-v)$ as $B_0$ and $B$ in
$\mathcal{L}(G)$ (and conversely), because this relation is forced
by the relation on the $Y$-shore.
\end{proof}
The next lemma shows instead that in case deletion of a vertex $v$ from a maximal biclique $B$ of $G$
causes loss of
maximality in the biclique, the role of $B$ in  $\mathcal{L}(G)$ is not really relevant.
\begin{lemma}\label{lem:Htree2} Let $G$ be a bipartite graph and
let $v\in V(G)$ and $B\in\maxbic(G)$ be such that $B-v\not\in\maxbic(G-v)$. If $v\in X$ and $B$ is not an atom in $\LL(G)$, then
$\deg^-_{\mathbf{H}(G)}(B)=1$. Moreover, if $(B',B)$ is the unique
arc entering $B$ in $\mathbf{H}(G)$ then $B'\in \maxbic(G-v)$.
Analogously, if $v\in Y$ and $B$ is not a co-atom in $\LL(G)$, then $\deg^+_{\mathbf{H}(G)}(B)=1$.
Moreover, if $(B,B')$ is the unique arc leaving $B$ in
$\mathbf{H}(G)$ then $B'\in \maxbic(G-v)$.
\end{lemma}
\begin{proof} By duality it suffices to prove the
statement only when $v\in X$. Since
$B$ is not an atom, $|X(B)| \geq 2$. Since $B-v\not\in\maxbic(G-v)$,
there is some $B'\in\maxbic(G-v)$ which dominates $B-v$.
Therefore, $X(B)-v\subseteq X(B')$ and $Y(B)\subseteq Y(B')$. It
follows that $B'\prec B$ in $\mathcal{L}^\circ(G)$ because
$Y(B)\subseteq Y(B')$ implies $B\nparallel B'$ and $B\not\prec
B'$. Consequently $X(B')\not=X(B)$ and $X(B')\subseteq X(B)$
implying that $X(B)-v=X(B')$. We therefore conclude that
$B'=(X(B)-v)\cup Y(B')$ and that $(B',B)$ is an arc of
$\mathbf{H}(G)$ with $B'\in\maxbic(G-v)$. Moreover, no other
arc $(B'',B)$ for some $B''\in\maxbic(G-v)$ can exist in
$\mathbf{H}(G)$. Indeed, if such an arc existed then all of the
following conditions would hold true:
\begin{itemize}
\item[--] $X(B')$ and $X(B'')$ are
inclusionwise incomparable because $B'\parallel B''$,
$\mathbf{H}(G)$ being a transitive reduction;
\item[--]  $v\not\in X(B'')$ because $B''\in\maxbic(G-v)$;
\item[--] $X(B'')\subseteq X(B)$ because $B''\prec B$.
\end{itemize}
The latter two conditions imply that $X(B'')\subseteq X(B)-v$, but
since $X(B)-v=X(B')$ the first one would be contradicted.
\end{proof}
\noindent Using standard terminology, as in \cite{bs}, a maximal biclique $B$
that satisfies the hypothesis of Lemma \ref{lem:Htree2} corresponds either
to a meet irreducible or to a join irreducible concept in the context associated to the bipartite graph.

The results of Lemmas \ref{lem:galois2bis}, \ref{lem:Htree1}, and \ref{lem:Htree2} imply:
\begin{theorem}\label{thm:retract}
Let $G$ be a BDH graph and let $v\in
V(G)$. Then one of the following conditions holds:
\begin{enumerate}
	\item\label{case:disconnect} $\mathbf{H}(G-v)$ has more connected components than $\mathbf{H}(G)$;
	\item\label{case:induced} $\mathbf{H}(G-v)$ is an induced subgraph of $\mathbf{H}(G)$;
	\item\label{case:contract} $\mathbf{H}(G-v)$ is a contraction of $\mathbf{H}(G)$.
\end{enumerate}
\end{theorem}
\begin{proof}
Let $\maxbic_0(v)\subseteq \maxbic(G)$ be the set of maximal bicliques $B$ containing $v$ such
that $B-v\not\in \maxbic(G-v)$. If $v$ is a cut-vertex of $G$, then condition \ref{case:disconnect} holds.
Otherwise, by Lemma \ref{lem:Htree1} and Lemma \ref{lem:Htree2},
$\mathbf{H}(G-v)$ can be derived from $\mathbf{H}(G)$ by the following
operations:
\begin{itemize}
\item[--] if $v\in X$ and $B\in \maxbic_0(v)$ delete $B$ if it is a sink
in $\mathbf{H}(G)$, otherwise contract the unique arc $(B',B)$
with $B'\in\maxbic(G-v)$ to the single node $B'$;
\item[--] if $v\in Y$ and $B\in \maxbic_0(v)$ delete $B$ if it is a
source in $\mathbf{H}(G)$, otherwise contract the unique arc
$(B,B')$ with $B'\in\maxbic(G-v)$ to the single node $B'$.
\end{itemize}

\noindent In both cases, either condition \ref{case:induced} or condition \ref{case:contract} holds.
\end{proof}

\noindent\textbf{Proof of the \emph{only if part} of Theorem
\ref{thm:main}.} Let us assume that $\mathbf{H}(G)$ is a tree and
let us prove that $G$ is a BDH graph. By Theorem
\ref{thm:retract}, it follows in particular that if $G_0$ is an
induced connected subgraph of $G$, then $\mathbf{H}(G_0)$ is a
contraction of $\mathbf{H}(G_1)$ for some connected induced subgraph
$G_1$ of $G$ such that $G_0$ is an induced subgraph of $G_1$.
Hence, $\mathbf{H}(G_0)$ is a tree, being the contraction of some
subtree of $\mathbf{H}(G)$. Now, to establish the thesis, it
suffices to observe that if $G_0$ is either a domino or a
chordless cycle with length greater than four then
$\mathbf{H}(G_0)$ is not a tree (see Figure~\ref{fi:h}).\qed

\section{Encoding $\LL(G)$}\label{sec:encoding}
In this section, we show how the Galois lattice of a BDH graph can be realized as the
containment relation among directed paths in an arborescence. The
results are achieved by further exploiting the interplay between
BDH graphs and series-parallel graphs. To this end we first recall
the classical two equivalent characterizations of series-parallel
graphs.

\begin{theorem}\label{thm:spchar} Let $S$ be a 2-connected graph with at least two vertices and not isomorphic to $K_2$.
Then the following statements are equivalent and characterize
series-parallel graphs.
\begin{enumerate}[{\rm (a)}]
\item $S$ does not contain a homeomorphic copy of $K_4$, i.e., the complete graph on four
vertices;
\item $S$ can be recursively constructed starting from a \emph{digon} by
either adding an edge with the same end-vertices as an existing
one or subdividing an existing edge by the insertion of a new
vertex.
\end{enumerate}
\end{theorem}
Recall that a \emph{digon} is a graph formed by two edges with the
same end-vertices. It is trivially a series-parallel graph. The
operations described in Statement (b) of Theorem \ref{thm:spchar}
are referred to as \emph{parallel extension} and \emph{series
extension} respectively. Statement (a) in Theorem \ref{thm:spchar}
is Duffin's characterization by forbidden minors. As we are going
to show, the close resemblance between Bandelt and Mulder's
construction (Statement (ii) of Theorem \ref{thm:bdhchar}) and
Statement (b) above is not merely formal. To this end we need the
notion of \emph{fundamental graph} of a graph  which we briefly
recall here. In a connected graph a co-tree is the subgraph
spanned by the complement of the edge-set of a spanning tree. If
$T$ is a spanning tree of $S$ its co-tree is denoted by
$\overline{T}$. Given a connected undirected graph $S$ and one of
its spanning trees $T$, the \emph{fundamental graph of} $S$ is the
bipartite graph $G_S(T)$ with color classes $E(T)$ and
$E(\overline{T})$ where there is an edge between $e\in E(T)$ and
$f\in E(\overline{T})$ if $e\in C(f,T)$, $C(f,T)$ being the
edge-set of the unique cycle in the graph spanned by $E(T)\cup
\{f\}$. Such a cycle is the so called \emph{fundamental cycle
through} $f$ \emph{with respect to} $T$. It can be shown that if
$S$ is 2-connected then $G_S(T)$ is connected. Moreover, $G_S(T)$
does not determine $S$ in the sense that non-isomorphic graphs may
have isomorphic fundamental graphs. As the fundamental graph is a
matroid theoretical tool we refer the interested reader to
\cite{jeelenetal,truemper} for more details. Now we just need to
recast the effect of series and parallel extensions on a graph $S$
on its \emph{fundamental graph} with respect to a given tree and
to observe that adding pending vertices and twins are counterparts
of the above operations. These effects are summarized in the
following table.
\begin{table}[h]\label{table:spbdh}
\centering
\begin{tabular}{lll}
Operation on $S$ & & Operation on $G_S(T)$\\ \hline
Parallel extension on $x\in X$ &$\leftrightarrow$& adding a pending vertex in $Y$ adjacent to $x$ \\
Series extension on $x\in X$ & $\leftrightarrow$ & adding a twin of $x$ in $X$\\
Parallel extension on $y\in Y$ & $\leftrightarrow$ & adding a twin
of $y$ in $Y$\\ Series extension on $y\in Y$ & $\leftrightarrow$ &
adding a pending vertex in $X$ adjacent to $y$.\\ \hline
\end{tabular}
\caption{The effects of series and parallel extension on $S$ on
its fundamental graph $G_S(T)$ with color classes $X=E(T)$ and
$Y=E(\overline{T})$.}
\protect
\noindent\hrulefill
\end{table}

\noindent The following result is now just a remark.
\begin{theorem}\label{thm:spbdh} A connected bipartite graph $G$ with color classes $X$ and $Y$
and at least two vertices is a BDH graph if and only if it is the
fundamental graph of a 2-connected series-parallel graph.
\end{theorem}
\begin{proof}
The \emph{if part} is proved by induction on the order of $G$. The
assertion is true when $G$ has two vertices because $K_2$ is a BDH
graph and at the same time is also the fundamental graph of a
digon. Let now $G$ have $n\geq 3$ vertices and assume that the
assertion is true for BDH graphs with $n-1$ vertices. By Bandelt
and Mulder's construction (Statement (ii) of Theorem
\ref{thm:bdhchar}) $G$ is obtained from a BDH graph $G'$ either by
adding a pending vertex or a twin. Let $S'$ be a series-parallel
graph having $G'$ as fundamental graph with respect to some
spanning tree. Since, by Table 1, the latter two operations
correspond to series or parallel extension of $S'$, the result
follows by Statement (b) of Theorem \ref{thm:spchar}. Conversely, let $G$
be the fundamental graph of a series-parallel graph $S$ with
respect to some tree $T$. By Statement (b) of Theorem \ref{thm:spchar} and
Table 1, $G$ can be constructed starting from a single edge by
either adding twins or pending vertices. Therefore, $G$ is a BDH
graph by Bandelt and Mulder's construction (Statement (ii) of
Theorem \ref{thm:bdhchar}).
\end{proof}
As credited by Syslo \cite{syslo}, Shinoda, Chen, Yasuda,
Kajitani, and W. Mayeda, proved that series-parallel graphs can be
completely characterized by a property of their spanning trees.
They proved that every spanning tree of a series-parallel graph
$S$ is a depth-first search tree of a 2-isomorphic copy of $S$,
where 2-isomorphism of graphs (in the sense of Whitney
\cite{whitney}) is isomorphism of binary vector spaces between
cycle-spaces of graphs. We can avoid to enter details of such
notions and we can content ourselves of restating in our
terminology a direct consequence of the result.
\begin{theorem}[S. Shinoda et al., 1981; Syslo,
1984]\label{thm:syslo} Let $G$ be the fundamental graph of a
2-connected graph with color classes $X$ and $Y$ with, say, $X$
being the edge-set of a spanning tree. Then there exist a graph
$S'$, a spanning tree $T'$ of $S'$ and an orientation $\phi $ of
$S'$ such that
\begin{itemize}
\item[--] $G\cong G_{S'}(T)$,
\item[--] $\phi T'$ is an arborescence,
\item[--] for each $x\in X$, $\{\phi z \ |\ z\in \{x\}\cup N(x)\}$ is the arc-set of a
directed circuit in $\phi S'$ and, consequently, $\{\phi z \ |\
z\in N(x)\}$ is the arc-set of a directed path in $\phi T'$,
\end{itemize}
if and only if $G$ is the fundamental graph of a series-parallel
graph.
\end{theorem}
Recall that an \emph{arborescence} is a directed tree with a
single special node distinguished as the \emph{root} such that,
for each other vertex, there is a dipath from the root to that
vertex. Syslo himself gave a constructive algorithmic proof of the
above result \cite{syslo}. A bipartite graph $G$ satisfying the
third condition of Theorem \ref{thm:syslo} will be called a
\emph{path-arborescence} bipartite graph and the arborescence
$\phi T'$ whose existence is asserted in the theorem will be
referred to as a \emph{supporting arborescence}. In general such
an arborescence will not be unique. Remark that series-parallel
graphs form a self-dual class of planar graphs, therefore Theorem
\ref{thm:syslo} holds simultaneously for the color class $Y$, $Y$
being the edge-set of a co-tree. We can state now the following
straightforward consequence of Theorem \ref{thm:spbdh} and Theorem
\ref{thm:syslo}.

\begin{corollary}\label{cor:dpt}
If $G$ is a connected BDH graph with color classes $X$ and $Y$, then $G$ is a \emph{path-arborescence} bipartite graph.
\end{corollary}
Let now $G$ be a BDH graph with color classes $X$ and $Y$. By Corollary \ref{cor:dpt}, there exists an
arborescence $\phi T$ with root $r$ and $X=E(T)$ that supports $G$ and such that, for each $y\in Y$, the set $\{\phi x \ |\ x\in N(y)\}$ is the arc set of a directed path in $\phi T$. Moreover, for $x \in X$, $N(x)$ is a set of dipaths each containing the arc $\phi x$. We now show that the inclusion-wise maximal such paths along with their pairwise intersections give the containment relation of the second coordinate of bicliques in $\mathcal{L}^\circ(G)$, which is in turn isomorphic to $\mathcal{L}^\circ(G)$. This fact allows an efficient encoding of the Galois lattice.

Notice that $\phi$ naturally induces a partial order $\leq_T$ on $X$ (the
\emph{arborescence order}) where $x\leq_T x'$ if
$\phi x$ is an arc of each dipath containing the root and $x'$.
Obviously, dipaths are intervals in this order and conversely.
Denote by $[xx']\subseteq X$ the set of elements in the interval
defined by the dipath having $\phi x$ and $\phi x'$ as end-arcs.
We also say that a subset $Z$ of $X$ spans a dipath if $\phi Z$ is
a dipath in $\phi T$. 
We show that the Galois lattice of a BDH graph is completely determined by some pairwise intersections of neighborhoods, plus some simple neighborhoods.
\begin{corollary}\label{cor:encoding}
Let $G$ be a connected BDH graph with color classes $X$ and $Y$.
Let $$\mathcal{F}=\big\{ N(x)\cap N(x') \ |\ x\not=x',\,\,\, x,x'\in X\big\}\cup \big\{ N(x),\, x\in X\big\}.$$ Then
$$\mathcal{L}^\circ(G)\cong \big(\mathcal{F},\subseteq\big).$$
\end{corollary}
\begin{proof} We show that $\big\{Y(B) \ |\ B\in \mathcal{L}^\circ(G)\big\}=\mathcal{F}$ and this is enough to prove the result because by \eqref{eq:shoresisodx}, 
$$\mathcal{L}^\circ(G)\cong \Big(\big\{Y(B) \ |\ B\in \mathcal{L}^\circ(G)\big\},\subseteq\Big),$$
Observe in the first place that\footnote{Formula \eqref{eq:ganterwille1} is \emph{concept polarity} in~\cite{gw}.}
\begin{equation}\label{eq:ganterwille1}
\big(X_0,Y_0\big)\in \mathcal{L}^\circ(G)\Longleftrightarrow 
\left\{
\begin{array} {lcl}
X_0 & = & \bigcap_{y\in Y_0}N(y)\\
Y_0 & = & \bigcap_{x\in X_0}N(x)
\end{array}
\right.
\end{equation}
Let $(X_0,Y_0)$ be a maximal biclique of $\mathcal{L}^\circ(G)$, and let $p_0 = |X_0|$. We first show that if $X_0=\bigcap_{y\in Y_0}N(y)$ then $Y_0\in \mathcal{F}$, i.e., either $Y_0 = N(x)$ for some $x \in X$ or $Y_0 = N(x) \cap N(x')$ for some $\{x,x'\} \subseteq X$. Let $\phi T$ be a supporting arborescence such that $N(x)$ is mapped onto a path in $\phi T$, for each $x \in X$. Since each $N(x)$ spans a dipath in $\phi T$, and since the intersections of dipaths always is a dipath, it follows that $\phi Y_0$ is the arc-set of some nonempty dipath $P$ of $\phi T$. Let $[a_i, b_i]$, with $a_i \leq_T b_i$, be the dipath spanned by $N(x_i)$, for each $x_i \in X_0$. We observe that if there exist two end-arcs $a_i$ and $a_j$, with $1 \leq i < j \leq p_0$, such that $a_i \parallel_T a_j$, then $N(x_i) \cap N(x_j) = \emptyset$.
Thus, since $Y_0$ is not empty, $\leq_T$ defines a total order on $a_i$'s, because all the $a_i$'s are pairwise comparable. 
Let $a_{i_M}$ be the maximum w.r.t. $\leq_T$ among all $a_i$'s, where $[a_{i_M}, b_{i_M}]$ is the path spanned by $N(x_{i_M})$.
Moreover, $ a_{i_M} \leq_T b_j$ for each $j$, otherwise it would be $N(x_{i_M}) \cap N(x_j) = \emptyset$ and $\bigcap_{x \in X_0}N(x)= \emptyset$. Hence, for each $1 \leq j \leq p_0$ there exists an arc $b_j' = b_{i_M} \wedge_T  b_j$ that is the maximum element in $[a_{i_M},b_{i_M}] \cap [a_j,b_j]$. Since $b_j' \in [a_{i_M},b_{i_M}]$, for each $1 \leq j \leq p_0$, then $\leq_T$ defines a total order on $b_j'$, for $1 \leq j \leq p_0$.
Let  $b_{i_m}'$ be the minimum among $b_j'$, for $1 \leq j \leq p_0$. It follows that
$\bigcap_{x\in X_0}N(x) = N(x_{i_M}) \cap N(x_{i_m})$, since $N(x_{i_M}) \cap N(x_j) \supseteq N(x_{i_M}) \cap N(x_{i_m})$ for each $1 \leq j \leq p_0$, where it can also happen that $x_{i_M} = x_{i_m}$.
We conclude that $Y_0\in \mathcal{F}$.

Let us prove, conversely, that if $Y_0\in \mathcal{F}$, then $Y_0=\bigcap_{x\in X_0}N(x)$ for some  $X_0\subseteq X$. This fact implies $Y_0\in\{Y(B) \ |\ B\in \mathcal{L}^\circ(G)\}$ by \eqref{eq:ganterwille1}.
\comment{

\textbf{inizio vecchia parte}

Since $Y_0\in \mathcal{F}$, then either $Y_0 = N(x)$, for some $x \in X$, or $Y_0 = N(x_1) \cap N(x_2)$, for some $\{x_1, x_2\} \subseteq X$. In the first case $(\mbox{twins}(x), Y_0)$ is a maximal biclique, where $\mbox{twins}(x) = \left\{x' \in X\ |\ N(x') = N(x)\right\}$.
In the second case, we show that for each pair $\{x_1, x_2\} \subset X$, there is always a maximal biclique $(X', N(x_1) \cap N(x_2))$, with $X' \supseteq \{x_1, x_2\}$. In fact, $(\{x_1, x_2\}, Y_0)$ is a biclique, and any biclique $(X', Y')$ that dominates $(\{x_1, x_2\}, Y_0))$, since $X' \supseteq \{x_1, x_2\}$, must have $Y' \subseteq N(x_1) \cap N(x_2) = Y_0$, hence $Y' = Y_0$. Thus, there must be a maximal biclique $(X',Y_0)$, for some $X'\supseteq \{x_1, x_2\}$.
This fact completes the proof.
\textbf{fine vecchia parte}
}
This is proved by a more general argument, in fact  (see also~\cite{gw}, page 19) for each set $X' \subseteq X$, with $X' \not= \emptyset$, there exists a maximal biclique
$(X'', \bigcap_{x \in X'}N(x))$ with $X'' \supseteq X'$. Let $Y_0 = \bigcap_{x \in X'}N(x)$: we have that
$(X', Y_0)$ is a biclique, and any biclique $(X'', Y')$ that dominates $(X', Y_0))$, since $X'' \supseteq X'$, must have $Y' \subseteq Y_0$, hence $Y' = Y_0$. Thus, there must be a maximal biclique $(X'',Y_0)$, for some $X''\supseteq X'$.
This fact completes the proof.
\end{proof}
The proof of Corollary~\ref{cor:encoding} shows that $\LL(G)$ is isomorphic to the containment relation of a set of paths in an arborescence. This implies that $\LL(G)$ has dimension at most 3. The bound derives from the following slightly more general consideration:
\begin{prop}
The containment order among paths in an arborescence has dimension at most 3.
\end{prop}
\begin{proof}
We show that the containment order among paths in an arborescence $T$ is a subposet of poset $(T, \leq_T) \times (\{1,2,\ldots,\tau\}, \geq)$, where $\tau$ is the height of $T$ and $\geq$ is the restriction of the natural order of integers to $\{1,2,\ldots,\tau\}$. We associate each path $\pi$ in $T$ with the pair $(e(\pi),d(\pi))$, where $e(\pi)$ is the maximum edge in $\pi$ with respect to $\leq_T$, and $d(\pi)$ is the distance of $\pi$ from the root of $T$, i.e., $d(\pi)$ is 1 if $\pi$ starts from the root, $d(\pi)$ is 2 if $\pi$ starts from a child of the root, and so on.
It is immediate to see that $\pi$ contains $\pi'$ if and only if $e(\pi) \leq_T e(\pi')$ and $d(\pi) \geq d(\pi')$.

It is well known that the arborescence order $\leq_T$ has dimension at most 2 (see \cite{trotter}), therefore the dimension of the product is at most $2+1$. 
\end{proof}
Thus, we can state that:
\begin{corollary}\label{coro:dim}
If $G$ is a BDH graph, then $\LL^\circ(G)$ has dimension at most 3.
\end{corollary}

Since $\LL^\circ(G)$ is a tree-like poset if $G$ is a BDH graph, Corollary~\ref{coro:dim} can also be obtained directly from a result by Trotter and Moore (see \cite{trmoo3}), asserting that a tree-like poset has dimension at most 3.

Notice that, there are containment orders among paths in an arborescence that are not isomorphic to the Galois lattice of any BDH graph. For example, the Galois lattice of a domino is isomorphic to the containment among sets $\{a,b\},\{b,c\},\{a,b,c\}$, and it is immediate to see that these sets are the edge sets of three subpaths of a path with edges $a,b,c$, which is clearly an arborescence.

\section{Efficiently computing (maximal) bicliques}\label{sec:algorithms}
In this section we discuss some of the algorithmic consequences of the encoding described in Section~\ref{sec:encoding}, and exploited in the proof of Corollary~\ref{cor:encoding}. By the results of \cite{wagner}, there exists an algorithm that given a
BDH graph computes a supporting arborescence $\phi T$ for $G$ as in Corollary~\ref{cor:dpt}.
The
algorithm runs in almost linear time in the size of $G$, that is
in time $O(\alpha(|X|,m)\cdot m)$ where $m$ is the number of edges of
$G$ and $\alpha$ is an inverse of the Ackermann function, which grows very slowly
and behaves essentially as a small constant even for very large values
of its arguments.
We propose a compact encoding of the BDH graph that requires $O(n)$ space in the worst case, where $n$ is the order of $G$. The retrieval of the neighborhood of any vertex requires linear time in the size of the neighborhood. Moreover, intersection of neighborhoods can be listed in optimal linear time in the size of the intersection, in the worst case.

At the same time each vertex $x_i \in X$ can be associated to a pair of edges $a_i, b_i$ so that $N(x_i)$ is mapped into the path from $a_i$ to $b_i$ in $\phi T$.
Recall that the partial order $\leq_T$ has linear dimension 2, so each vertex/edge in the arborescence can be equipped with a pair of labels in $\{1,\ldots,n\}$ so that relation $\leq_T$ between two edges is verified in constant time. After fixing an arbitrary ordering on the outgoing edges for each vertex in $T$, the pair of labelings is defined by two preorder numberings of $T$, one obtained by visiting at each vertex outgoing edges from left to right and the other one obtained by visiting at each vertex outgoing edges from right to left.

This gives an encoding of $G$ that allows to answer the following queries in optimal worst case time, where $x \in X$ and $X' \subseteq X$:
\begin{enumerate}
\item\label{algo:neigh} list $N(x)$, in time $O(|N(x)|)$;
\item\label{algo:check} check whether $\bigcap_{x \in X'} N(x) = \emptyset$, in time $O(|X'|)$;
\item\label{algo:list} list $\bigcap_{x \in X'} N(x)$, in time $O\left(|X'| + \left|\bigcap_{x \in X'} N(x)\right|\right)$;
\item\label{algo:maximal} check whether $\left(X', \bigcap_{x \in X'} N(x)\right)$ is a maximal biclique, in $O\left(|X'| + \left|\bigcap_{x \in X'} N(x)\right|\right)$ worst case time.
\end{enumerate}
\begin{remark}
In lattice theoretical terminology, query~\ref{algo:maximal} corresponds to checking whether $A = A^{\ast\ast}$, where $^\ast$ denotes concept polarity as defined in~\cite{gw}.  
\end{remark}

Note that the size of the encoding is only $O(n)$, while the number of edges in a BDH graph can be $\Theta(n^2)$, and still allows the computation of the maximal biclique containing a given set $X'$ on one side in time linear in the in the number of vertices in the biclique.

The algorithm to solve query \ref{algo:list} (queries \ref{algo:neigh} and \ref{algo:check} are special cases of query \ref{algo:list}) is described in Figure~\ref{fi:intersection}, and follows the same argument as in the first part of the proof of Corollary~\ref{cor:encoding}. Let $X' = \{x_1, x_2, \ldots, x_k\}$, and let $(a_i, b_i)$, for $1 \leq i \leq k$, be the end-arcs of the path associated to $x_i$ in $T$.

\begin{figure}[ht]
\noindent\hrulefill%
\begin{prog}{pr:resilientspanner}
Given $X' \subseteq X$, compute $\bigcap_{x \in X'} N(x)$.\\
We assume the arborescence $T$ is given, and a data structure for solving lowest common ancestor queries\\
according to $\leq_T$, as described in \cite{schieber}, has been built.\\
$(a_i, b_i)$, for $1 \leq i \leq |X'|$, are the end-arcs of the path in $T$ associated to $N(x_i)$\\
\vspace{0.3cm}\\
\N  \key{let}$a_{\mathrm{max}} = a_1$\\
\N  \key{for}$i=2$ to $|X'|$\\
\NL{line:magg}  \>  \key{if}$a_i >_T a_{\mathrm{max}}$\\
\N  \>  \>  \key{let}$a_{\mathrm{max}} = a_i$\\
\NL{line:nomagg}  \>  \key{else if}$a_i \not\leq_T a_{\mathrm{max}}$\\
\N  \>  \>  \key{return}$\emptyset$\\
\NL{line:lca}  \key{let}$b_{\mathrm{min}} = \bigwedge_T\{b_1, b_2, \ldots, b_k\}$\\
\N  \key{if}$b_{\mathrm{min}} \geq_T a_{\mathrm{max}}$\\
\NL{line:path}  \>  \key{return}$[a_{\mathrm{max}},b_{\mathrm{min}}]$\\
\N  \key{else}\\
\N  \>  \key{return}$\emptyset$\\
\end{prog}

\caption{Algorithm \texttt{NeighborIntersection}.}
\protect\label{fi:intersection}
\noindent\hrulefill%
\end{figure}

It can be seen that the complexity of  algorithm \texttt{NeighborIntersection} is\\ $O\left(|X'| + \left|\bigcap_{x \in X'} N(x)\right|\right)$, since tests in Lines \ref{line:magg} and \ref{line:nomagg} are performed in constant time starting from the encoding of the 2-dimensional partial order $\leq_T$. The computation of the lowest common ancestor $\bigwedge_T$ at line \ref{line:lca} is computed in time $O(|X'|)$ using the data structure proposed in \cite{schieber}, which is built in $O(n)$ time.

Path retrieval in Line \ref{line:path} requires $O\left(\left|\bigcap_{x \in X'} N(x)\right|\right)$ worst case time, starting from $b_{\mathrm{min}}$ and following parent pointers in the arborescence $T$ up to $a_{\mathrm{max}}$.

In order to solve query \ref{algo:check}, we can still use algorithm \texttt{NeighborIntersection}, without listing the path in Line~\ref{line:path}, thus requiring $O(|X'|)$ worst case time. Query~\ref{algo:maximal} can be solved using the same algorithm, thanks to Observation (\ref{eq:ganterwille1}), provided that the same encoding is stored both for side $X$ and for side $Y$. In fact, $\left(X', \bigcap_{x \in X'} N(x)\right)$ is a maximal biclique if and only if $X' = \bigcap_{y\in Y_0} N(y)$, where $Y_0 = \bigcap_{x\in X'}N(x)$, that can be checked by computing $Y_0$ and then computing $\bigcap_{y\in Y_0} N(y)$, i.e., solving two queries of type \ref{algo:list}.

\section{An indirect proof of Theorem \ref{thm:main}}\label{sec:indirect}
We prove here Theorem \ref{thm:main} by exploiting existing results, that is by taking the longest path between \textbf{BDH} and $\mathbb{T}_\B$ in Diagram (\ref{eq:diagram2}). To this end we need some more terminology on hypergraphs and Ptolemaic graphs. The reader is referred to the monographs \cite{berge} and \cite{bralespi}.
\mybreak

\paragraph{Hypergraphs related to bipartite chordal graphs}
Let $\HH$ be a hypergraph on $V$. Given two arbitrary linear orders of $V$ and $\HH$, let $\mv{A}(\mathcal{H})=\{a_{i,j}\}$ be the  $\{0,1\}^{m\times n}$-matrix whose rows correspond (in the order chosen for $\HH$) to the members of $\mathcal{H}$, the columns correspond (in the order chosen for $V$) to the vertices of $\mathcal{H}$ and where $a_{i,j}=1$ if the $i$-th element of $V$ is in the $j$-th member of $\HH$ and $a_{i,j}=0$ otherwise. Clearly, if $\mv{A}$ is a $\{0,1\}$-matrix, we can reverse the construction by associating with $\mv{A}$ the hypergraph $\HH(\mv{A})$ on the index-set $V$ of the columns and whose members are the supports of the rows of $\mv{A}$, regarded as subsets of $V$. 
A \emph{clutter} is a hypergraph whose members are inclusion-wise incomparable. The hypergraph $\widehat{\HH}$ is the collection obtained by closing $\mathcal{H}$ under intersection, namely, $F\in \widehat{\HH}$ if and only if either $F\in \HH$ or $F$ is the intersection of two or more members of $\HH$. By $\HH^\uparrow$ we denote the clutter obtained from $\HH$ as follows: first pairwise equal members are identified into a unique member and then only the inclusion-wise maximal members are retained.
If $\mathcal{F}$ is a family of subsets of a given common ground set, let $\mathcal{F}^\uparrow$ denote the collection consisting of the inclusion-wise maximal members of $\mathcal{F}$.
\mybreak
If $G$ is a bipartite graph with color classes $X$ and $Y$ we associate with $G$ the two hypergraphs $\NN_X(G)$ and $\NN_Y(G)$ on $X$ and $Y$, respectively, called \emph{neighborhood hypergraphs of $G$}, given by $\NN_X(G)=(N(y), \ |\ y\in Y)$ and $\NN_Y(G)=(N(x), \ |\ x\in X)$. The hypergraphs $\NN_X(G)^\uparrow$ and $\NN_Y(G)^\uparrow$ are called the \emph{maximal neighborhood systems} of $G$.
\mybreak
Let $G$ be a graph. The clutter $\K(G)$ is the clutter consisting of the maximal cliques of $G$. Hence $G\mapsto \K(G)$ induces a map $\tau$ that sends isomorphism classes of graphs into isomorphism classes of clutters. 
On the other hand with any hypergraph $\HH$ on $V$ we can associate the graph $(\HH)_2$ with vertex set $V$ and where two vertices $u$ and $v$ are joined by an edge if there is a member $F$ of $\HH$ containing both. The graph $(\HH)_2 $ is called the \emph{2-section of $\HH$}. As shown by the clutter $\HH=\{\{1,2\},\{2,3\},\{1,3\}\}$ the map $\tau$ is not in general a bijection, because $(\HH)_2\cong K_3$ and hence $\K((\HH)_2)\not\cong\HH^\uparrow$. This motivates the following notion: a hypergraph $\HH$ is \emph{conformal} if the maximal cliques of its 2-section coincide with the maximal members of $\HH$, that is $\K((\HH)_2)\cong\HH^\uparrow$ \cite{berge}. It follows that if $\tau$ is a bijection when its image is restricted to the class of conformal clutters.
\mybreak
If $\mv{R}$ and $\mv{S}$ are two $\{0,1\}$ matrices, we say that $\mv{S}$ contains a copy of $\mv{R}$ if the rows and the columns of $\mv{S}$ can be permuted so that the permuted matrix contains $\mv{R}$ as a submatrix. For an integer $h\geq 3$, let $\mv{C}_h=\big\{c_{i,j}\big\}\in \{0,1\}^{h\times h}$ be the matrix whose entries satisfy $c_{i,j}=1\Leftrightarrow i\equiv j\imod{h}\,\,\text{or}\,\,i+1\equiv j\imod{h}$. 
We now collect some very well known facts about \emph{totally balanced} hypergraphs, namely, those hypergraphs whose matrix is $\mv{C}_h$ free for all $h\geq 3$. All the characterizations listed below can be found in \cite{bralespi}.
\begin{enumerate}[(a)]
\item\label{com:a} $\HH$ is a totally balanced hypergraph, i.e., $\mv{A}(\HH)$ is $\mv{C}_h$-free for $h\geq 3$, if and only if $(\HH)_2$ is a strongly chordal graph. 
\item\label{com:b} $G$ is a bipartite chordal graph if and only if its neighborhood hypergraphs are totally balanced. 
\end{enumerate} 
\begin{remark}\label{rem:tbneigh}
Using the fact that totally balanced hypergraphs are conformal, we obtain
the following fact: \emph{Let $G$ be a bipartite graph with color classes $X$ and $Y$. Then $G$ is chordal if and only if each of the 2-sections of the following hypergraphs is strongly chordal: $\NN_X(G)$, $\NN_X(G)^\uparrow$, $\NN_Y(G)$ and $\NN_Y(G)^\uparrow$}.
\end{remark}
Let now $\C(G)=(\widehat{\mathcal{K}}(G)\cup\{\emptyset,V(G)\},\subseteq)$. Then $\C(G)$ is a lattice known as the \emph{clique lattice of $G$}. 
\mybreak
We need the following definition.
\begin{definition}\label{def:gammaacy}
A $\gamma$-acyclic hypergraph is a totally balanced hypergraph whose matrix does not contain a copy of the matrix 
\begin{equation}\label{eq:dominomatrix}
\mv{F}=\begin{pmatrix}
1 &0 &1\\
1 &1 &1\\
0 &1 &1
\end{pmatrix}
\end{equation}
\end{definition}

The Bachman Diagram of $\mathcal{H}$, denoted by \text{Bachman}($\mathcal{H}$), is the transitive reduction of the poset $(\widehat{\HH},\subseteq\big)$. 
Fagin \cite{fagin} proved the following
\begin{theorem}[\cite{fagin}]\label{thm:fagin}
$\mathcal{H}$ is $\gamma$-acyclic $\Longleftrightarrow$ $\widehat{\HH}$ is $\gamma$-acyclic $\Longleftrightarrow$ {\rm Bachman($\HH$)} is a tree.
\end{theorem}

\paragraph{Ptolemaic graphs}Ptolemaic graphs, as shown by Howorka \cite{how1,how2}, are precisely chordal distance hereditary graphs. There is an intimate relationship between Ptolemaic graphs and BDH graphs, made explicit by Bandelt and Mulder \cite{bm}, and Peled and Wu \cite{peledwu}. In particular, Bandelt and Muller show that if $G$ is a BDH graph then the graph obtained by completing certain level sets is distance hereditary and chordal and thereby Ptolemaic by Howorka's characterization. On the other hand, as mentioned above, Wu's result asserts that the vertex-clique graph of a Ptolemaic graph, is a BDH graph.
\mybreak
Besides their own theoretical importance, Ptolemaic graphs (and hence BDH graphs), deserve a special role in the theory of relational database, as shown by D'Atri and Moscarini in \cite{damo}. In the early eighty, deep investigations of theoretical properties of relational databases \cite{fagin,yanna} led to a refinement of the notion of cycles in hypergraphs yielding various degrees of acyclicity \cite{fagin}. Among them, the notion of $\gamma$-acyclic hypergraphs relates directly with Ptolemaic graphs in that the 2-section graph of $\gamma$-acyclic hypergraph is a Ptolemaic graph \cite{damo}. Since the Bachman diagram of a $\gamma$-acyclic hypergraph $\mathcal{H}$, namely the hypergraph obtained by closing the edge-set of $\mathcal{H}$ under intersection, is the clique-lattice of its 2-section graph with top and bottom removed, and since by the results of Fagin \cite{fagin}, the Bachman diagram of a hypergraph is a tree if and only if the hypergraph is $\gamma$-acyclic, it follows that the clique-lattice of a Ptolemaic graph is tree-shaped. Ueheara and Uno \cite{ueuno}, obtained the same result from another perspective:
they proved that cliques of a Ptolemaic graph have the remarkable property of being \emph{laminar}. Recall that a family of sets is said to be \emph{laminar} if given any two sets of the family, then either such two sets are disjoint or they are inclusion-wise comparable. If $\F$ is a laminar family, then $(\F,\subseteq)$ is a tree-like poset---this is another way of stating a classical results of Edmonds and Giles\cite{Edgi}--. In view of Fagin's result and the result of Uehara and Uno, we conclude that the clique lattice of a Ptolemaic graph is tree shaped.
\mybreak

D'Atri and Moscarini \cite{damo} elicited the relation between $\gamma$-acyclicity and Ptolemaicity as follows: 
\begin{theorem}[\cite{damo}]\label{thm:damo}
$G$ is Ptolemaic $\Longleftrightarrow$ $\K(G)$ is $\gamma$-acyclic.
\end{theorem}
After Theorem \ref{thm:damo}, the following result of Uehara and Uno, though discovered independently, is readily seen to be equivalent to Fagin's result
\begin{theorem}[\cite{ueuno}]\label{thm:ueuno}
$G$ is Ptolemaic $\Longleftrightarrow$ $\widehat{\K(G)}$ is laminar $\Longleftrightarrow${\rm Bachman($\K(G)$)} is a tree.
\end{theorem}
To see how Theorems~\ref{thm:fagin} and~\ref{thm:ueuno} imply, via Diagram~(\ref{eq:diagram2}), Theorem~\ref{thm:main}, we proceed as follows: we first prove that the maps $\lambda$ and $\mu$ that make Diagram \ref{eq:diagram2} commuting exist. They are indeed obtained in Proposition~\ref{prop:mu} by specializing the characterization given in remark \ref{rem:tbneigh} to the subclasses $\mathbf{Pt}$ and $\mathbf{BDH}$. Then we show in Theorem \ref{thm:bridge} that $\lambda$ and $\mu$ actually make the diagram commuting. To this end we need some intermediate results on Galois lattices of BDH graphs. 
\begin{prop}\label{prop:mu}
If $G$ is Ptolemaic graph then $\Gamma(\K(H))$ is a BDH graph. If $G$ is a BDH graph with color classes $X$ and $Y$, then each of the 2-sections of the following hypergraphs is Ptolemaic: $\NN_X(G)$, $\NN_X(G)^\uparrow$, $\NN_Y(G)$ and $\NN_Y(G)^\uparrow$.
\end{prop}
\begin{proof}
The first part is Wu's result \cite{peledwu}. Let us prove the second part and let $G_0$ be an induced subgraph of $G$ with color classes $X_0$ and $Y_0$. Let $\mv{A}_0$ be the $\{0,1\}$-matrix with rows indexed by $X_0$ columns indexed by $Y_0$ defined by the adjacency of the vertices in $X_0$ and $Y_0$. Observe that $\mv{A}_0$ is a submatrix of both $\mv{A}(\NN_X(G))$ and $\mv{A}(\NN_Y(G))$ (this is just a matter of checking definitions). Since $G$ is BDH it is domino-free. Up to a permutation of rows and columns, the adjacency matrix of a domino is the matrix $\mv{F}$ defined in \eqref{eq:dominomatrix}. Hence $\NN_X(G)$, $\NN_X(G)^\uparrow$, $\NN_Y(G)$ and $\NN_Y(G)^\uparrow$ are $\gamma$-acyclic and therefore the 2-sections are Ptolemaic by Theorem \ref{thm:damo} (because of the conformality of $\gamma$-acyclic hypergraphs).
\end{proof}
After the proposition we see that the maps $\lambda$ and $\mu$ are defined as follows
\begin{subequations}
\begin{equation}\label{eq:mapsX}
\begin{array}{ccc}
\lambda:\mathbf{Pt}\rightarrow \mathbf{BHD} & &\mu_1:\mathbf{BHD}\rightarrow \mathbf{Pt}\\ 
\,\, G\mapsto \Gamma(\K(G)) & & \,\,\, H\mapsto (\NN_Y(H))_2
\end{array} 
\end{equation}
or 
\begin{equation}\label{eq:mapsY}
\begin{array}{ccc}
\lambda:\mathbf{Pt}\rightarrow \mathbf{BHD} & &\mu_2:\mathbf{BHD}\rightarrow \mathbf{Pt}\\ 
\,\, G\mapsto \Gamma(\K(G)) & & \,\,\, H\mapsto (\NN_X(H))_2
\end{array} 
\end{equation}
\end{subequations}
provided that $H$ is in $\mathbf{BDH}$ and has color classes $X$ and $Y$. The following lemma uses Lemmata \ref{lem:Htree1} and \ref{lem:Htree2} that are proved in Section \ref{sec:onlyif}.
\begin{lemma}\label{lem:paoloproof}
Let $G$ be a BDH graph with color classes $X$ and $Y$. If $x\in X$ is such that $N(x)\not\in \NN_Y(G)^\uparrow$ and $N(x)\not\in\widehat{\NN_Y(G)}$, then $\mathbf{H}(G-x)$ arises from $\mathbf{H}(G)$ by contracting some arcs.
the same results holds if $y\in Y$ satisfies the hypotheses. 
\end{lemma}
\begin{proof}
Let $\widetilde{\B}(x)=\{B\in \B(G) \ |\ B-x\not\in \B(G-x)\}$. By the hypotheses, $\{x\}$ is not the $X$-shore of any biclique in $\B(G)$, and therefore no member of $\widetilde{\B}(x)$ is an atom in $\LL(G)$ (by Lemma~\ref{lem:galois2bis}). Since $N(x)\not\in\widehat{\NN_Y(G)}$, it follows that $\LL(G)\not\cong \LL(G-x)$. By Lemma \ref{lem:Htree1}, $\LL(G-x)$ is a sub-poset (not necessarily a sub-lattice) of $\LL(G)$. The vertices of $\mathbf{H}(G-x)$ correspond bijectively to vertices of $\mathbf{H}(G)-\widetilde{\B}(x)$. However, by Lemma \ref{lem:Htree2}, the vertices in $\widetilde{\B}(x)$ are meet-irreducible in $\LL(G)$ and are covered in $\LL(G)$ by bicliques of  $\B(G)\setminus\B(x)$. In other words, for each $B\in \B(x)$ there is a unique arc $(B',B)$ entering $B$ and with $B'\not\in \B(x)$. Therefore, by contracting every such arc yields $\mathbf{H}(G-x)$. By duality we obtain the statement for the $Y$-shore. In the latter case however meet-irreducible is replaced by join-irreducible.
\end{proof}
The following theorem sets a bridge between the clique lattice of a graph $G$ and the Galois lattice of the vertex-clique bipartite graphs of $G$ and is strong enough to make Diagram \eqref{eq:diagram2} commuting. For a bipartite graph $G$ with color classes $X$ and $Y$, let $I_X$ and $I_Y$ be the set of vertices in $X$ and $Y$ respectively, satisfying the hypotheses of Lemma~\ref{lem:paoloproof}. After Lemma \ref{lem:paoloproof}, the following fact is straightforward.
\begin{lemma}\label{lem:tree-tree}
Let $G$ be a BDH graph with color classes $X$ and $Y$. Then $\LL(G)$ is tree-shaped if and only if both $\LL(G-I_X)$ and $\LL(G-I_Y)$ are tree-shaped.
\end{lemma}

\begin{theorem}\label{thm:bridge}
If $G$ is Ptolemaic graph, then there exists a lattice isomorphism $\Phi_\lambda$ such that $\C(G)\cong \LL(\Gamma(\K(G)))=\LL(\lambda G)$. If $G$ is a BDH graph with color classes $X$ and $Y$, then there are lattice isomorphisms $\Psi_{\mu_1}$ and $\Psi_{\mu_2}$ such that $\LL(G-I_X)\cong \C((\NN_Y(G))_2)=\C(\mu_1 G)$ and $\LL(G-I_Y)\cong \C((\NN_X(G))_2)=\C(\mu_2 G)$.
\end{theorem}
\begin{proof}
Let $G$ be Ptolemaic, let $H=\Gamma(\K(G))$. Suppose that $S\in \widehat{\K}(G)$. If $\K(S)$ denotes the set of all maximal cliques of $G$ containing $S$, then $(S,\K(S))$ is a biclique of $\B(H)$ and $S$ is the $X$-shore of a biclique in $\B(H)$. Conversely, If $(S,\K')\in \B(H)$, for some $S\subseteq V(G)$ and some $\K'\subseteq \K(G)$, then $S$ is precisely the intersection of the cliques in $K'$ and therefore it is in $\widehat{\K}(G)$. In other words the isomorphism $\Phi_\lambda$ is the one given in \eqref{eq:shoresisosx}.
Let $H_1=(\NN_Y(G))_2$ and $H_2=(\NN_X(G))_2$. Since $G$ is BDH it is bipartite chordal and therefore $\NN_Y(G)$ and $\NN_X(G)$ are both conformal (see also Remark~\ref{rem:tbneigh}). Hence $\K(H_1)=\NN_Y(G)^\uparrow$ and $\K(H_2)=\NN_X(G)^\uparrow$. Clearly $\Gamma(\K(H_1))\cong G-I_X$ and $\Gamma(\K(H_2))\cong G-I_Y$ where $\cong$ is graph isomorphism. Now $\LL(G-I_X)\cong \C(H_1)$ and $\LL(G-I_Y)\cong \C(H_2)$, where now $\cong$ is lattice isomorphism, and the two isomorphisms $\Psi_{\mu_1}$ and $\Psi_{\mu_2}$ are explicitly given by \eqref{eq:shoresisodx} and \eqref{eq:shoresisosx} with $G-I_X$ and $G-I_Y$ in place of $G$, respectively.
\end{proof}
We are almost done. To obtain a proof Theorem \ref{thm:main}, it is now sufficient to resort to Theorem \ref{thm:bridge}, to observe that 
by Proposition \ref{prop:mu}, both $(\NN_Y(G))_2$ and $(\NN_X(G))_2$ are Ptolemaic, to invoke either Fagin's or Uehara and Uno's result and finally to apply Lemma~\ref{lem:tree-tree}.


\begin{thebibliography}{12}

\bibitem{avg}
J. Amilhastre, M.~C. Vilarem, P. Janssen. \emph{Complexity of
minimum biclique cover and minimum biclique decomposition for
bipartite domino-free graphs}. Discrete Applied Mathematics, 86
(1998), 125--144.

\bibitem{abs}
R. Arratia, G. Bollob\'{a}s, G. Sorkin. \emph{The interlace
polynomial: a new graph polynomial}. Proceedings of the eleventh
annual ACM-SIAM symposium on discrete algorithms, San Francisco,
CA, Jan. 2000, 237--245.

\bibitem{atk}
M.~D. Atkinson. \emph{On computing the number of linear extensions
of a tree}. Order, 7 (1990) 23--25.

\bibitem{audamo}
G. Ausiello, A. D'Atri, M. Moscarini. \emph{Chordality properties
on graphs and minimal conceptual connections in semantic data
models}. J. Comput. System Sci., 33 (1986) 179--202.

\bibitem{bm}
H.~J. Bandelt, H.~M. Mulder. \emph{Distance-hereditary graphs}. J.
Combin. Theory B, 41 (1986) 182--208.

\bibitem{bdov}
R. Belohl{\'a}vek, G. De Baets, J. Outrata, V. Vychodil. \emph{Trees in concept lattices}.
Proceedings of the 4th Int.
Conf. on Modeling Decisions for Artificial Intelligence,  MDAI 2007,
Kitakyushu, Japan, Aug. 2007, Volume 4617 of LNCS, Springer, 174--184.

\bibitem{berge}
C.~Berge, Graphs and hypergraphs, {\em North-Holland, Amsterdam} (1973).


\bibitem{bs}
A. Berry, A. Sigayret. \emph{Dismantlable lattices in the mirror}.
Proceedings of the 11th International Conference on Formal Concept Analysis,
ICFCA 2013, Dresden, Germany, May 2013.
Volume 7880 of LNCS, Springer, 44--59.

\bibitem{bralespi}
A. Brandst\"{a}dt, V.~G. Le, J.~P. Spinrad. \emph{Graph classes: a
survey}. SIAM Monographs on Discrete Mathematics and Applications,
Society for Industrial and Applied Mathematics (SIAM),
Philadelphia, PA, 1999.

\bibitem{bg}
F. Brucker, A. G\'ely.
\emph{Crown-free lattices and their related graphs}. Order, 28(3) (2011) 443--454.

\bibitem{codiste}
S. Cornelsen, G. Di Stefano. \emph{Treelike comparability graphs}.
Discrete Applied Mathematics, 157 (2009) 1711--1722

\bibitem{damo}
A.~D'Atri, M.~Moscarini. \emph{On hypergraph acyclicity and graph chordality}. Inform. Proc. Lett. 29 (1988) 271--274.

\bibitem{duffin}
R.~J. Duffin. \emph{Topology of series-parallel networks}. J.
Math. Analysis Appl., 10 (1965) 303--318.

\bibitem{Edgi}
J.~Edmonds, R.~Giles. \emph{A min-max relation for submodular functions
on graphs}. Studies in Integer Programming, Proceedings of Workshop on
Programming, Bonn, 1975. Ann. Discrete Math. 1 (1977) 185--204.


\bibitem{mosa}
J.~A. Ellis-Monaghan, I. Sarmiento. \emph{Distance hereditary
graphs and the interlace polynomial}. Combinatorics, Probability
\& Computing, 16:6 (2007) 947--973.

\bibitem{fagin}
R.~Fagin.\emph{Degrees of acyclicity for hypergraphs and relational database schemes}. J. ACM,30, 3 (1983) 514--550.

\bibitem{gw}
G. Ganter. R. Wille.
\emph {Formal concept analysis - mathematical foundations},
Springer, (1999).

\bibitem{jeelenetal}
J.~F. Geelen, A.~M.~H. Gerards, A. Kapoor. \emph{The excluded
minors for $GF(4)$-representable matroids}. J. Comb. Theory, Ser.
G 79:2 (2000) 247--299.

\bibitem{how1}
E. Howorka. \emph{A characterization of distance-hereditary
graphs}. Quart. J. Math. Oxford Ser., 2:26 (1977) 417--420.

\bibitem{how2}
E. Howorka. \emph{A characterization of Ptolemaic graphs, survey
of results}, in Proceedings of the 8th SE Conf. Combinatorics,
Graph Theory and Computing, (1977) 355--361.

\bibitem{peledwu}
U~N. Peled, J. Wu. \emph{Restricted unimodular chordal graphs}.
J. of Graph Theory, 30:2 (1999), 121--136.

\bibitem{rival}
I. Rival. \emph{Lattices with doubly irreducible elements}.
Canadian Mathematical Bulletin,
17:1 (1974) 91--95.

\bibitem{schieber}
G. Schieber, U. Vishkin. \emph{On finding lowest common ancestors: simplification and parallelization}. SIAM Journal on Computing 17:6  (1988) 1253--1262.

\bibitem{wagner}
R.~P. Swaminathan, D.~B. Wagner. \emph{The
arborescence-realization problem}. Discrete Applied Mathematics,
59 (1995) 267--283.

\bibitem{syslo}
M.~M. Syslo. \emph{Series-parallel graphs and depth-first search
trees.} IEEE Transactions on Circuits and Systems, 31:12 (1984)
1029--1033.

\bibitem{trotter}
W.~T. Trotter. \emph{Combinatorics and partially ordered sets: dimension theory}.
The Johns Hopkins University Press, Baltimore, Maryland (1992).

\bibitem{trmoo3}
W.~T. Trotter, J. Moore. \emph{The dimension of planar posets}. J.
Combin. Theory B, 21 (1977) 51--67.

\bibitem{truemper}
B. Truemper. \emph{Matroid decomposition}. Academic Press, Boston
(1992).

\bibitem{ueuno}  
R.~Uehara, Y.~Uno. \emph{Laminar structure of Ptolemaic graphs with applications}. Discrete Applied Mathematics, 157:7 (2009) 1533--1543.

\bibitem{whitney}
H.~Whitney. \emph{2-isomorphic graphs}. Amer. Math. J., 55 (1933)
245--254.

\bibitem{yanna}
M.~Yannakakis. \emph{Algorithms for acyclic database schemes}. In Proc. 7th Int. Conf. on Very Large
Databases (Cannes, France, Mar. 29-31, 1982), ACM, New York, (1982) 82--94.

\end{thebibliography}
\end{document}